\documentclass[review]{elsarticle}
\usepackage[english]{babel}
\usepackage{amssymb}

\usepackage{lipsum}
\usepackage{caption}
\usepackage{mwe}
\usepackage{subfig}
\usepackage{graphicx}
\usepackage{color}
\usepackage{float}
\usepackage{multicol}
\usepackage{morefloats}
\usepackage{amsthm,amssymb,amsmath,amsfonts}

\newtheorem{theorem}{Theorem}
\theoremstyle{definition}
\newtheorem{definition}{Definition}[section]
\usepackage{etoolbox}
\makeatletter
\patchcmd{\ps@pprintTitle}{\footnotesize\itshape
	Preprint submitted to \ifx\@journal\@empty Elsevier
	\else\@journal\fi\hfill\today}{\relax}{}{}
\makeatother

\begin{document}
\begin{frontmatter}
	
	\title{On nontrivially knotted closed trajectories and positive
		lower bound for the energy}

	\author{Kaveh Mohammadi}
	\ead{kmsmath@aut.ac.ir}


	\begin{abstract}
		In this paper we suggest a soluiton to this question Does the existence of
		any nontrivially knotted closed trajectory for a field presents a positive lower
		bound for the energy and hence blocks the relaxation of the field into arbi-
		trarily small energies[1]
	\end{abstract}
	
	
\end{frontmatter}



\section{introduction}
In this paper we suggest a soluiton to this question Does the existence of
any nontrivially knotted closed trajectory for a field presents a positive lower
bound for the energy and hence blocks the relaxation of the field into arbitrarily
small energies

One of the main question in the area of topological fluid dynamics which still stays unsolved is this question does the existence of any nontrivially knotted closed trajectories yield a positive lower bound for the energy hence intercept the relaxation of the filed to arbitrarily small energies\cite{boris2005}.

Fluid inside stars is fully conducting, in 1942 Hanns Alfven assert this idea that magnetic field lines are frozen inside fluid and they move with the movement of the fluid, this idea is well known as Alfven frozen-in theorem.ownig to Maxwell's equation conducting fluid inside the stars continuous to move  up until surplus of energy is entirely disappeared .the mutual linking of the field line (trajectories) intercepts the full disappearance of the magnetic energy.~\cite{Ricca}
\par 
I this paper we will suggest a conjecture to this challenging question.our conjecture is inspired and motivated by the work  of Sakharov and Zeldovich. They had proposed a solution for a divergence-free filed which trajectories of it are all closed and pairwise unlinked.and their solution was proved twenty years later 
by M.fredman\cite{boris2005,both}
\par 
This paper is organized a follows in the second part we review some important concept  .in the third part we state our conjecture to the challenging question. 
\section{Some important concept}
\begin{theorem}
If $M_1$ and $M_2$ are two compact manifold in $\mathbb{R}^n$ which posses smooth boundaries,they have the same volume moreover are diffeomorphic hence there exist a volume presrving diffeomorphims $\phi M_1\rightarrow M_2$, then satisfying $\phi^*\Omega=\Omega$\cite{EduardZehnder}
\end{theorem}
\begin{definition}
A knot is a subspace of $\mathbb{R}^3$ which is homeomorphic to the circle.
there are different types of knots,the simplest knot is called unknot or trivial knot which is a cricle embeded in $\mathbb{R}^3$ and the simplest example of nontrivial knots are terfoil knot ,figure eight knot.\cite{Colin}
\end{definition}
On a three dimensional Euclidean space the helicity of a magnetic field attributes a number to the linking of the field line with one another, moreover, helicity computes the typical linking of the field line or their mutual bonding.
V.arnold in 1973 defined the helicity without the use of metric in other words he defined helicity in terms of a topological propetry which preserved under the action of volume-preserving diffeomorphism \cite{both, Arnold}
\section{Statement of the  theorem}
Assume that the magnetic field is confined to the object (Fig.1) in which field line (trajectories) are nontrivially linked and, assume that field vanishes outside the object. to decrease the energy in a significant way, the length of most trajectories need to be lessened. After applying the suitable diffeomorphism the length of most trajectories will be lessened, however, because of the action of volume preserving diffeomorphism, the shape of the object will be deformed into a fat object with the same volume.
moreover, the mutual linking of the field lines(trajectories) will be preserved too. consequently, the energy of the magnetic field is decreased significantly.
\begin{theorem}
The energy of a magnetic field inside a 3-dimensional object can be decreased significantly by the action of a suitable diffeomorphism that preserves the volume and mutual linking of the field line (trajectories).
\end{theorem}

\begin{proof}

~\cite{zeldovich}
\cite{both}.In order to restore the boundary condition, several actions need to be applied to the object(boundary constraint is removed in the case of the star). first, we divide the object $C$ into a small shell $A_\epsilon$ of thickness b and $C_{1-\epsilon}$ a smaller object inside the shell. The shell $A_\epsilon$ is extremely stretchable when the smaller object inside is removed it will be squeezed into a smaller shell $A_{\epsilon^\prime}$.when the smaller-object is removed, we will squeeze it in a way that it can be fit inside the smaller shell $A_\epsilon$ the action which turned $C_{1-\epsilon}$ into $C_{1-\xi^\prime}$ is an ambient $C^\infty$.after applying the action, we have a diffeomorphism $f:A_\epsilon\rightarrow A_{\epsilon^\prime} $ and the push forward field $f_\star B=B_1$ its energy can be expressed as 
\[\int_{A_\epsilon}||B||^2=\int_{A_{\epsilon^\prime}}||B_1||^2+\int_{G}||B_1||^2 \]

such action of squeezing shrinks the $\xi$-orbits inside the $C_{1-\xi}$ to a desired prescribed factor, therefore it reduces the field energy in the (transformed) smaller object $C_1-\xi$ to the desired level. squeezing $C_1-\xi$ into $C_{1-\xi^\prime}$ and embedding inside $A_\xi^\prime$ preserves the mutual bonding of the field trajectories and also reduces the energy because trajectories have been reduced to a great extent.

Now it is time to estimate the energy of the field in the image shell $A_\xi^\prime$ by applying theorem.1 in order to control the maximal squeezing of orbits in the shell, it suffices to provide boundaries of the squeezing of the volume element for an arbitrarily thin shell. 
\begin{itemize}
	\item[(a)]first expand a thin shell
	\item [(b)]lastly, squeeze this neighborhood to K.
	
\end{itemize}
As thikness approaches zero the enegry $A_\xi^\prime$ tend to zero because energy integrated is bounded independtly of b.if b is chosen sufficiently small, so $A_\xi^\prime$$=\approx \xi $. defining the initial squeezing of the sub-object $C_1-\xi$ into a smaller sub-object $C_1-\xi^\prime$ of length $\ell$.necessary  to organize a family $\varphi _t$ of diffeomorphism sso the area of every horizontal section is squeezed by the factor of t as a result total energy will be reduced to $C(\varphi^\star_t\xi)\approx \frac{1}{t}$. hence, the supremum norm $||\varphi^\star_t||_{\ell^\infty}=max ||\xi_t||=\mathcal{O}(t)$

\end{proof}

\begin{figure}[H]
\centering
\begin{multicols}{2}
	\includegraphics[scale=0.3]{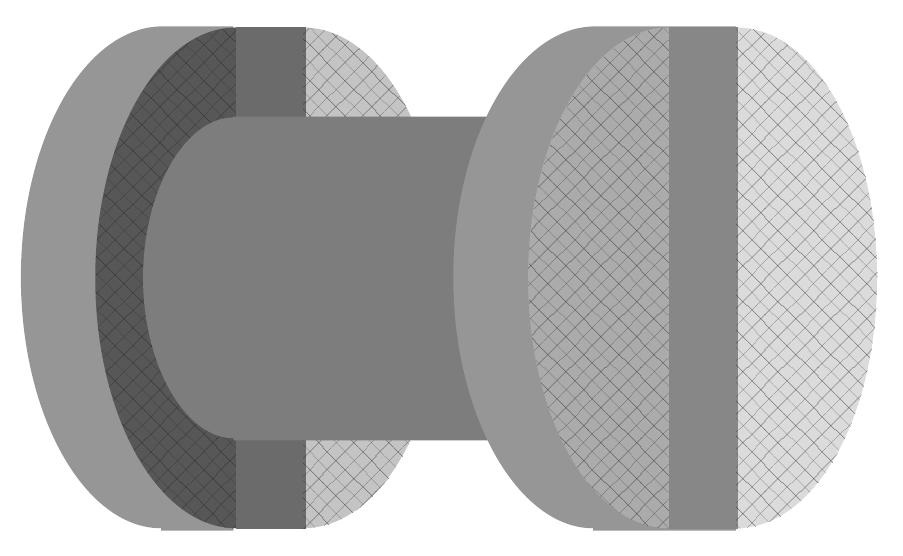}
	\vspace{10mm}\\
	
	\includegraphics[scale=0.35]{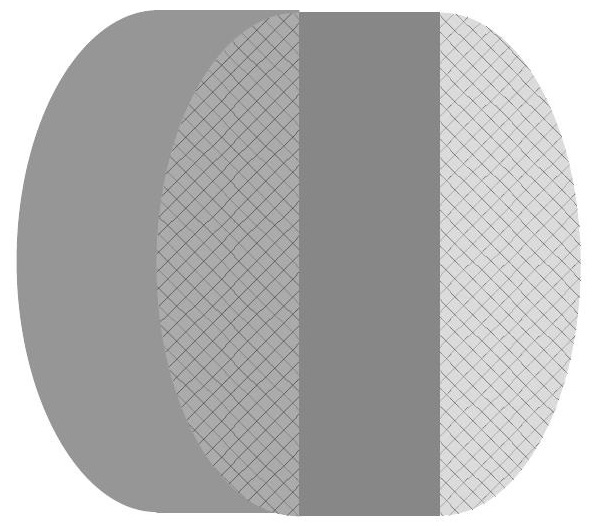}
	
	\includegraphics[scale=0.3]{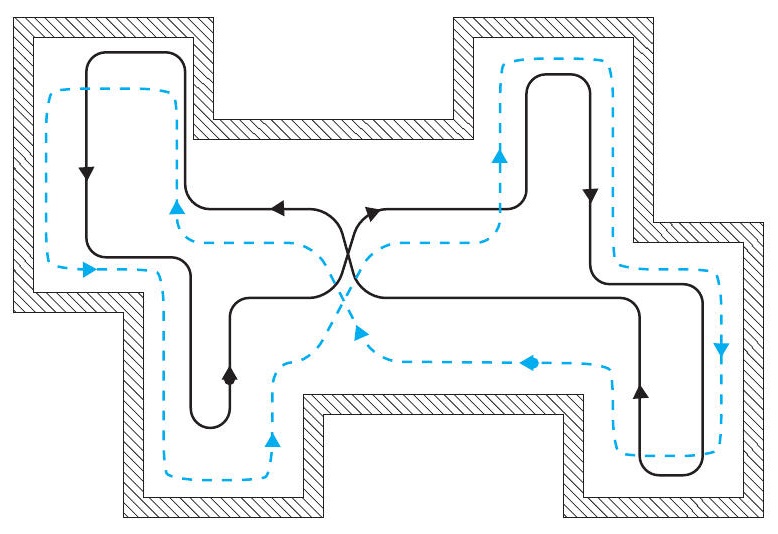}
	\vspace{20mm}\\
	
	\includegraphics[scale=0.3]{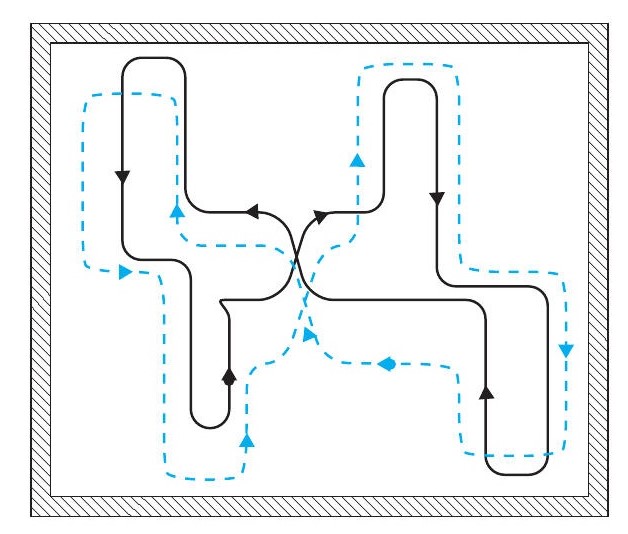}
\end{multicols}
\caption{(a)3 dimensional object which will be turn ito (b)by a an action of volume preserving diffeomorphism}
\label{3fig}
\end{figure}

\section*{References}

\bibliography{mybibfile}

\end{document}